\pdfoutput=1
\documentclass[11pt,letterpaper]{article}

\usepackage[english]{babel}
\usepackage[utf8]{inputenc}
\usepackage[T1]{fontenc}
\usepackage{lmodern}
\usepackage[letterpaper,textwidth=6.0in,textheight=8.6in,centering]{geometry}
\usepackage{amsmath,amsthm,amssymb,mathtools}
\usepackage{booktabs}
\usepackage{microtype}
\usepackage{enumitem}
\usepackage{array}
\usepackage{tikz}
\usetikzlibrary{arrows.meta,positioning,bending,calc}
\usepackage{algorithm}
\usepackage{algpseudocode}
\usepackage[colorlinks=true,linkcolor=blue!50!black,citecolor=blue!50!black,urlcolor=blue!50!black]{hyperref}

\hypersetup{
  pdftitle={Classifying Approximation Algorithms: Understanding the APX Complexity Class},
  pdfauthor={Bruce Changlong Xu and Arthur Lee},
  pdfsubject={Approximation Algorithms and Computational Complexity (cs.CC)}
}

\theoremstyle{plain}
\newtheorem{theorem}{Theorem}[section]
\newtheorem{lemma}[theorem]{Lemma}
\newtheorem{proposition}[theorem]{Proposition}
\newtheorem{corollary}[theorem]{Corollary}

\theoremstyle{definition}
\newtheorem{definition}[theorem]{Definition}

\theoremstyle{remark}
\newtheorem{remark}[theorem]{Remark}

\newcommand{\OPT}{\operatorname{OPT}}
\newcommand{\VAL}{\operatorname{VAL}}
\newcommand{\ALG}{\operatorname{ALG}}
\newcommand{\poly}{\operatorname{poly}}
\newcommand{\prob}[1]{\textnormal{\textsc{#1}}}
\newcommand{\NPO}{\mathsf{NPO}}

\newcommand{\APX}{\mathsf{APX}}
\newcommand{\PTAS}{\mathsf{PTAS}}
\newcommand{\FPTAS}{\mathsf{FPTAS}}
\newcommand{\MaxSNP}{\mathsf{MaxSNP}}
\newcommand{\NP}{\mathsf{NP}}
\newcommand{\Pclass}{\mathsf{P}}
\newcommand{\BPP}{\mathsf{BPP}}

\newcommand{\coRP}{\mathsf{coRP}}
\newcommand{\FP}{\mathsf{FP}}
\newcommand{\OptP}{\mathsf{OptP}}
\newcommand{\logAPX}{\mathsf{log\text{-}APX}}
\newcommand{\polyAPX}{\mathsf{poly\text{-}APX}}

\title{\textbf{Classifying Approximation Algorithms:\\[0.15em]
Understanding the $\APX$ Complexity Class}}

\author{
  Bruce Changlong Xu \qquad Arthur Lee \\[0.35em]
  {\normalsize Stanford University}
}

\date{}

\begin{document}
\maketitle

\begin{abstract}
\noindent Many $\NP$-hard optimization problems that are equivalent under Karp reductions behave very differently when we ask for polynomial-time approximations. This paper studies the complexity class $\APX \subseteq \NPO$, which consists of optimization problems that admit polynomial-time algorithms with constant-factor approximation guarantees. We first review the formal structure of $\NPO$, define L-reductions and PTAS-reductions, and prove that L-reductions preserve polynomial-time approximation schemes ($\PTAS$). Using the Probabilistically Checkable Proofs (PCP) theorem as a starting point, we present self-contained L-reductions establishing the $\APX$-completeness of \prob{Max-3SAT}, \prob{Max-2SAT}, \prob{Max-NAE-3SAT}, and \prob{Max-Cut}, and we situate $\APX$ within the broader hierarchy $\FPTAS \subseteq \PTAS \subseteq \APX \subseteq \logAPX \subseteq \polyAPX$. We then examine structural consequences of the polynomial hierarchy for self-improvability and $\APX$-intermediate problems such as \prob{Min-Bin-Packing}, and review how $\APX$-hardness arises in Bayesian optimal mechanism design for three-bidder single-item auctions. Turning to algorithm design, we analyze three core techniques across the $\APX$ and $\logAPX$ boundary: linear programming relaxation with Raghavan--Thompson randomized rounding for multicommodity congestion minimization, Goemans--Williamson semidefinite programming rounding for \prob{Max-Cut}, and greedy, dual-fitting, and randomized rounding methods for \prob{Min-Set-Cover}. We close by discussing bounded-query approximation oracles and recent constant-factor and sub-$3/2$ advances on the asymmetric and metric traveling salesman problems.

\vspace{0.65em}
\end{abstract}

\section{Introduction}
\label{sec:intro}

Under standard polynomial-time many-one reductions, all $\NP$-complete decision problems are computationally equivalent: a polynomial-time algorithm for any one of them yields a polynomial-time algorithm for all of $\NP$. That equivalence collapses once we move from exact decision problems to optimization and ask how closely a polynomial-time algorithm can approximate the optimal objective value.

A classic example is the contrast between \prob{Min-Vertex-Cover} and \prob{Max-Independent-Set} on an undirected graph $G = (V, E)$ with $n = |V|$ vertices. A subset $S \subseteq V$ is an independent set if and only if its complement $V \setminus S$ is a vertex cover, so the optimal values satisfy
\begin{equation}
  \OPT_{\prob{VC}}(G) + \OPT_{\prob{IS}}(G) = n.
  \label{eq:vc-is-complement}
\end{equation}
For exact optimization, finding a minimum vertex cover and finding a maximum independent set are identical. Under multiplicative approximation, however, the affine subtraction $n - \OPT_{\prob{VC}}(G)$ distorts ratios whenever $\OPT_{\prob{VC}}(G)$ is close to $n$. Taking the endpoints of any maximal matching gives a polynomial-time $2$-approximation for \prob{Min-Vertex-Cover}, placing it in $\APX$. By contrast, \prob{Max-Independent-Set} (or equivalently \prob{Max-Clique} on the complement graph) cannot be approximated within a factor of $n^{1-\varepsilon}$ for any $\varepsilon > 0$ unless $\Pclass = \NP$ \cite{Hastad1999,Zuckerman2007}.

The complexity class $\APX$ captures those $\NPO$ optimization problems for which some polynomial-time algorithm achieves a constant multiplicative approximation ratio on every instance. Its study grew out of the syntactic class $\MaxSNP$, introduced by Papadimitriou and Yannakakis \cite{Papadimitriou1991} using existential second-order logic to classify problems with constant-factor approximations, and was subsequently given a machine-based treatment in terms of $\NPO$ and PTAS-reductions by Crescenzi, Kann, Silvestri, and Trevisan \cite{Crescenzi1999} and Ausiello et al.\ \cite{Ausiello1999}; Khanna, Motwani, Sudan, and Vazirani \cite{Khanna1998} showed that the closure of $\MaxSNP$ under PTAS-preserving reductions coincides with $\APX$.

Several representative problems illustrate the internal structure of $\APX$:

\paragraph{Load balancing and load rebalancing.}
In makespan scheduling on $m$ identical parallel machines, we are given $n$ jobs with processing times $p_1, \dots, p_n \ge 0$ and must assign each job to a machine to minimize the maximum load $\max_{1 \le i \le m} \sum_{j \in A_i} p_j$. Graham's list scheduling achieves a $(2 - 1/m)$-approximation, so the problem lies in $\APX$; Hochbaum and Shmoys \cite{Hochbaum1987} showed by dual approximation and rounding that it further admits a polynomial-time approximation scheme ($\PTAS$), achieving ratio $1 + \varepsilon$ for every fixed $\varepsilon > 0$. Aggarwal, Motwani, and Zhu \cite{LoadRebalancing} proved that the \emph{load rebalancing} variant---where jobs are already assigned to machines and one seeks to improve the makespan subject to a relocation constraint---also lies in $\PTAS$.

\paragraph{Minimum vertex cover.}
Given a graph $G = (V, E)$, a vertex cover is a set $C \subseteq V$ containing at least one endpoint of every edge $e \in E$. Both maximal-matching rounding and linear programming (LP) threshold rounding achieve a $2$-approximation in polynomial time. In general graphs, Dinur and Safra \cite{DinurSafra2005} proved via the PCP theorem that approximating \prob{Min-Vertex-Cover} within $10\sqrt{5} - 21 \approx 1.3606$ is $\NP$-hard, and Khot and Regev \cite{KhotRegev2008} showed that under the Unique Games Conjecture no $(2 - \varepsilon)$-approximation exists for any $\varepsilon > 0$. Even on cubic graphs (maximum degree $3$), Alimonti and Kann \cite{Alimonti2000} proved that \prob{Min-Vertex-Cover} is $\APX$-complete.

\paragraph{Traveling salesman problem (metric and asymmetric).}
Given a complete graph $G = (V, E)$ with edge costs $c: E \to \mathbb{Q}_{\ge 0}$, the traveling salesman problem (\prob{TSP}) asks for a minimum-cost Hamiltonian tour visiting every vertex exactly once. Without the triangle inequality, \prob{TSP} cannot be approximated within any polynomial factor in polynomial time unless $\Pclass = \NP$, because a gap reduction from \prob{Hamiltonian-Cycle} assigns cost $1$ to graph edges and a large cost to non-edges. When costs obey the triangle inequality $c(u,w) \le c(u,v) + c(v,w)$ (\prob{Metric-TSP}), doubling a minimum spanning tree yields a $2$-approximation, and Christofides' algorithm \cite{Christofides1976} combines a minimum spanning tree with a minimum-weight perfect matching on its odd-degree vertices to achieve a $3/2$-approximation. After forty-five years at $3/2$, Karlin, Klein, and Oveis Gharan \cite{Karlin2021} broke the $3/2$ barrier for \prob{Metric-TSP} by random sampling from max-entropy spanning tree distributions over the subtour elimination LP, achieving a $(3/2 - \varepsilon_0)$-approximation for a constant $\varepsilon_0 > 10^{-36}$. For the directed asymmetric traveling salesman problem (\prob{ATSP}) with triangle inequality, where only $O(\log n)$ and $O(\log n / \log \log n)$ approximations were known for decades, Svensson, Tarnawski, and V\'egh \cite{Svensson2020} established the first constant-factor polynomial-time approximation algorithm, placing \prob{ATSP} inside $\APX$.

\paragraph{Maximum 3-satisfiability (\prob{Max-3SAT}).}
Given a Boolean formula $\varphi$ in conjunctive normal form with at most three literals per clause, \prob{Max-3SAT} asks for a truth assignment maximizing the number of satisfied clauses. A uniform random assignment satisfies $7/8$ of all exact $3$-clauses in expectation \cite{Johnson1974}, and H\aa stad \cite{Hastad2001} proved that achieving $(7/8 + \varepsilon)$ for any $\varepsilon > 0$ is $\NP$-hard. Consequently, \prob{Max-3SAT} is the canonical $\APX$-complete problem: it lies in $\APX$, and every problem in $\APX$ reduces to it via an approximation-preserving reduction.

\section{Optimization Classes and Approximation-Preserving Reductions}
\label{sec:definitions}

\subsection{NPO, APX, and the Approximability Spectrum}

To place approximation algorithms on a complexity-theoretic footing, we first formalize $\NP$ optimization problems \cite{Ausiello1999}.

\begin{definition}[$\NPO$ problem]
\label{def:npo}
An optimization problem $\Pi \in \NPO$ is a quadruple $\Pi = (\mathcal{I}_\Pi, \mathrm{Sol}_\Pi, \VAL_\Pi, \mathrm{goal}_\Pi)$ such that:
\begin{enumerate}[label=(\roman*),topsep=2pt,itemsep=2pt]
  \item $\mathcal{I}_\Pi \subseteq \{0,1\}^*$ is the set of valid instances, recognizable in polynomial time;
  \item for each instance $I \in \mathcal{I}_\Pi$, $\mathrm{Sol}_\Pi(I) \subseteq \{0,1\}^*$ is the set of feasible solutions of $I$; there exists a polynomial $p$ such that $|y| \le p(|I|)$ for all $y \in \mathrm{Sol}_\Pi(I)$, and the relation $\{(I, y) : I \in \mathcal{I}_\Pi,\, y \in \mathrm{Sol}_\Pi(I)\}$ is decidable in polynomial time;\footnote{Throughout this paper we assume that $\mathrm{Sol}_\Pi(I)$ is nonempty and that $\OPT_\Pi(I) > 0$ on non-trivial instances. For problems like \prob{Max-Clique}, any single vertex forms a feasible solution of value $1$.}
  \item $\VAL_\Pi(I, y) \in \mathbb{Q}_{>0}$ is a polynomial-time computable objective function assigning a positive value to every feasible solution $y \in \mathrm{Sol}_\Pi(I)$ (written $\VAL_\Pi(y)$ when $I$ is clear from context);
  \item $\mathrm{goal}_\Pi \in \{\min, \max\}$ specifies whether the objective is to be minimized or maximized, with optimal value
  \[
    \OPT_\Pi(I) = \mathrm{goal}_{y \in \mathrm{Sol}_\Pi(I)} \VAL_\Pi(I, y).
  \]
\end{enumerate}
\end{definition}

For a polynomial-time algorithm $\mathcal{A}$ that maps each instance $I \in \mathcal{I}_\Pi$ to a feasible solution $\mathcal{A}(I) \in \mathrm{Sol}_\Pi(I)$, we write $\ALG(I) = \VAL_\Pi(I, \mathcal{A}(I))$. To unify minimization and maximization, we define the performance ratio $R_\Pi(I, y) \ge 1$ of a feasible solution $y \in \mathrm{Sol}_\Pi(I)$ by
\begin{equation}
  R_\Pi(I, y) = \begin{cases}
    \dfrac{\VAL_\Pi(I, y)}{\OPT_\Pi(I)} & \text{if } \mathrm{goal}_\Pi = \min, \\[0.8em]
    \dfrac{\OPT_\Pi(I)}{\VAL_\Pi(I, y)} & \text{if } \mathrm{goal}_\Pi = \max.
  \end{cases}
  \label{eq:perf-ratio}
\end{equation}
In portions of the literature on maximization problems, the approximation ratio is alternatively written as $\VAL_\Pi(I,y) / \OPT_\Pi(I) \in (0, 1]$ (for example, saying that \prob{Max-3SAT} has a $7/8$-approximation rather than an $8/7$-approximation). Both conventions are equivalent when the constant is fixed; in formal definitions below we specify the direction explicitly.

\begin{definition}[$\APX$ and $f(n)$-$\APX$]
\label{def:apx}
An optimization problem $\Pi \in \NPO$ belongs to $\APX$ if there exist a polynomial-time algorithm $\mathcal{A}$ and a constant $c \ge 1$ such that for all instances $I \in \mathcal{I}_\Pi$,
\begin{equation}
  R_\Pi(I, \mathcal{A}(I)) \le c.
\end{equation}
More generally, for a function $f: \mathbb{N} \to [1, \infty)$, $\Pi \in f(n)\text{-}\APX$ if there exists a polynomial-time algorithm $\mathcal{A}$ satisfying $R_\Pi(I, \mathcal{A}(I)) \le O(f(|I|))$ for all $I \in \mathcal{I}_\Pi$. In particular, $\logAPX$ and $\polyAPX$ correspond to $f(n) = \log n$ and $f(n) = n^k$ for some constant $k > 0$.
\end{definition}

\begin{proposition}[\prob{Max-3SAT} $\in \APX$]
\label{prop:max3sat-in-apx}
Let \prob{Max-3SAT} be the problem of finding a truth assignment that maximizes the number of satisfied clauses in a $3$-CNF formula $\varphi$. Then $\prob{Max-3SAT} \in \APX$ with deterministic performance ratio at most $2$ (and at most $8/7$ when every clause contains three distinct variables).
\end{proposition}
\begin{proof}
Let $\varphi$ be a $3$-CNF formula with $m = |C|$ clauses (without complimentary literals inside the same clause). Consider the deterministic algorithm $\mathcal{A}$ that evaluates two truth assignments: the all-true assignment $\tau_{\mathbf{1}}(x_i) = 1$ for all $i$, and the all-false assignment $\tau_{\mathbf{0}}(x_i) = 0$ for all $i$, and outputs whichever assignment satisfies more clauses. Every clause $c \in C$ contains at least one literal; if $c$ contains a positive literal it is satisfied by $\tau_{\mathbf{1}}$, and if $c$ contains only negative literals it is satisfied by $\tau_{\mathbf{0}}$. Hence every clause is satisfied by at least one of $\tau_{\mathbf{1}}$ or $\tau_{\mathbf{0}}$, so the two assignments together satisfy at least $m$ clauses across the two runs. By averaging,
\[
  \ALG(\varphi) = \max\{\VAL(\tau_{\mathbf{1}}), \VAL(\tau_{\mathbf{0}})\} \ge \frac{m}{2} \ge \frac{1}{2}\OPT(\varphi),
\]
so $\OPT(\varphi)/\ALG(\varphi) \le 2$. Moreover, if each clause has $k$ distinct literals, setting each variable independently to $1$ with probability $1/2$ satisfies each clause with probability $1 - 2^{-k} \ge 1/2$ (and $7/8$ when $k=3$); derandomizing by the method of conditional expectations yields a deterministic algorithm achieving $\ALG(\varphi) \ge \frac{7}{8}m \ge \frac{7}{8}\OPT(\varphi)$ on exact $3$-CNF formulas \cite{Johnson1974}.
\end{proof}

\begin{definition}[$\PTAS$ and $\FPTAS$]
\label{def:ptas}
Let $\Pi \in \NPO$ be a minimization (respectively, maximization) problem. A \emph{polynomial-time approximation scheme} ($\PTAS$) for $\Pi$ is a family of algorithms $\{\mathcal{A}_\varepsilon\}_{\varepsilon > 0}$ such that for every fixed $\varepsilon \in (0, 1)$ and every instance $I \in \mathcal{I}_\Pi$, $\mathcal{A}_\varepsilon(I)$ runs in time $\poly(|I|)$ and produces a feasible solution satisfying
\[
  \VAL_\Pi(I, \mathcal{A}_\varepsilon(I)) \le (1 + \varepsilon)\OPT_\Pi(I)
  \quad\text{(resp.\ $\VAL_\Pi(I, \mathcal{A}_\varepsilon(I)) \ge (1 - \varepsilon)\OPT_\Pi(I)$)}.
\]
If the running time of $\mathcal{A}_\varepsilon(I)$ is polynomial in both $|I|$ and $1/\varepsilon$, then $\{\mathcal{A}_\varepsilon\}$ is a \emph{fully polynomial-time approximation scheme} ($\FPTAS$). When $\mathcal{A}_\varepsilon$ is a randomized polynomial-time algorithm achieving the $(1\pm\varepsilon)$ guarantee with probability at least $2/3$, the scheme is a \emph{polynomial-time randomized approximation scheme} ($\mathsf{PRAS}$), or $\mathsf{FPRAS}$ if polynomial in $|I|$ and $1/\varepsilon$.
\end{definition}

Just as $\NP$-intermediate problems lie in $\NP \setminus (\Pclass \cup \NP\text{-complete})$, $\APX$-intermediate problems lie in $\APX \setminus (\PTAS \cup \APX\text{-complete})$. A natural related example is \prob{Min-Bin-Packing}, which admits an \emph{asymptotic} $\PTAS$ producing at most $(1+\varepsilon)\OPT(I) + 1$ bins (Karmarkar--Karp), whereas distinguishing $\OPT(I) = 2$ from $\OPT(I) = 3$ is $\NP$-hard by reduction from \prob{Partition} (ruling out a true $\PTAS$ with ratio $< 3/2$), and as we discuss in Section~\ref{sec:ph}, \prob{Min-Bin-Packing} cannot be $\APX$-complete unless the polynomial hierarchy collapses \cite{Crescenzi1999}.

\subsection{L-Reductions and the Taxonomy of Approximation-Preserving Reductions}

Karp reductions do not preserve approximation ratios because they only map optimal decision thresholds without controlling objective values away from optimality. Papadimitriou and Yannakakis \cite{Papadimitriou1991} introduced \emph{linear reductions} (L-reductions) to ensure that both the optimal value and the absolute error of any proposed solution are linearly bounded across the reduction.

\begin{definition}[L-reduction \cite{Papadimitriou1991}]
\label{def:l-reduction}
Let $P_1, P_2 \in \NPO$ be two optimization problems. We say that $P_1$ is \emph{L-reducible} to $P_2$, written $P_1 \le_L P_2$, if there exist polynomial-time computable functions $f$ and $g$ and positive constants $a, b > 0$ such that for every instance $x \in \mathcal{I}_{P_1}$:
\begin{enumerate}[label=(\roman*),topsep=2pt,itemsep=2pt]
  \item $f(x) \in \mathcal{I}_{P_2}$ is a valid instance of $P_2$;
  \item for every feasible solution $y \in \mathrm{Sol}_{P_2}(f(x))$, $g(x, y) \in \mathrm{Sol}_{P_1}(x)$ is a feasible solution of $x$ (written $g(y)$ when $x$ is clear from context);
  \item the optimal objective values satisfy
  \begin{equation}
    \OPT_{P_2}(f(x)) \le a \cdot \OPT_{P_1}(x);
    \label{eq:l-red-opt}
  \end{equation}
  \item for every feasible solution $y \in \mathrm{Sol}_{P_2}(f(x))$, the absolute errors from optimality satisfy
  \begin{equation}
    |\OPT_{P_1}(x) - \VAL_{P_1}(g(y))| \le b \cdot |\OPT_{P_2}(f(x)) - \VAL_{P_2}(y)|.
    \label{eq:l-red-err}
  \end{equation}
\end{enumerate}
\end{definition}

It follows immediately from Definition~\ref{def:l-reduction} that L-reductions compose: if $P_1 \le_L P_2$ with constants $(a_1, b_1)$ and $P_2 \le_L P_3$ with constants $(a_2, b_2)$, then $P_1 \le_L P_3$ with constants $(a_1 a_2, b_1 b_2)$. Furthermore, L-reductions preserve membership in $\PTAS$:

\begin{lemma}[Preservation of $\PTAS$ under L-reductions]
\label{lem:l-red-ptas}
Let $P_1, P_2 \in \NPO$ with $P_1 \le_L P_2$. If $P_2 \in \PTAS$, then $P_1 \in \PTAS$.
\end{lemma}
\begin{proof}
Let $(f, g, a, b)$ be an L-reduction from $P_1$ to $P_2$, and fix any target accuracy $\varepsilon \in (0, 1)$. Set $\delta = \frac{\varepsilon}{ab} > 0$.

First suppose $P_1$ and $P_2$ are minimization problems, so $\VAL_{P_1}(g(y)) \ge \OPT_{P_1}(x)$ and $\VAL_{P_2}(y) \ge \OPT_{P_2}(f(x))$. Run the $(1 + \delta)$-approximation algorithm for $P_2$ on $f(x)$ to obtain a solution $y$ satisfying $\VAL_{P_2}(y) \le (1 + \delta)\OPT_{P_2}(f(x))$, or equivalently $\VAL_{P_2}(y) - \OPT_{P_2}(f(x)) \le \delta \OPT_{P_2}(f(x))$. Applying \eqref{eq:l-red-err} and \eqref{eq:l-red-opt} gives
\begin{align*}
  \frac{\VAL_{P_1}(g(y))}{\OPT_{P_1}(x)}
  &\le \frac{\OPT_{P_1}(x) + b \cdot \bigl(\VAL_{P_2}(y) - \OPT_{P_2}(f(x))\bigr)}{\OPT_{P_1}(x)} \\
  &= 1 + b \cdot \frac{\VAL_{P_2}(y) - \OPT_{P_2}(f(x))}{\OPT_{P_1}(x)} \\
  &\le 1 + a b \cdot \frac{\VAL_{P_2}(y) - \OPT_{P_2}(f(x))}{\OPT_{P_2}(f(x))} \\
  &\le 1 + a b \cdot \delta = 1 + \varepsilon.
\end{align*}
Similarly, if $P_1$ and $P_2$ are maximization problems, let $y$ be a $(1 - \delta)$-approximate solution for $f(x)$, so $\OPT_{P_2}(f(x)) - \VAL_{P_2}(y) \le \delta \OPT_{P_2}(f(x))$. Then \eqref{eq:l-red-err} and \eqref{eq:l-red-opt} yield
\begin{align*}
  \OPT_{P_1}(x) - \VAL_{P_1}(g(y))
  &\le b\bigl(\OPT_{P_2}(f(x)) - \VAL_{P_2}(y)\bigr) \\
  &\le b \delta \OPT_{P_2}(f(x))
  \le a b \delta \OPT_{P_1}(x)
  = \varepsilon \OPT_{P_1}(x),
\end{align*}
so $\VAL_{P_1}(g(y)) \ge (1 - \varepsilon)\OPT_{P_1}(x)$. Because $f$, $g$, and the $(1\pm\delta)$-approximation for $P_2$ run in time polynomial in $|x|$ for each fixed $\varepsilon$, $P_1 \in \PTAS$.
\end{proof}

\begin{remark}[L-reductions, E-reductions, and PTAS-reductions]
\label{rem:reduction-taxonomy}
While L-reductions always preserve $\PTAS$, an L-reduction between two maximization problems preserves constant-factor approximability ($\APX$) automatically when $\OPT_{P_2}(f(x))$ and $\OPT_{P_1}(x)$ are linearly related, whereas for general problems in $\NPO$ several closely related reductions have been studied \cite{Crescenzi1999,Khanna1998,Ausiello1999}:
\begin{itemize}[topsep=2pt,itemsep=2pt]
  \item \textbf{Strict reducibility:} $R_{P_1}(x, g(x,y)) \le R_{P_2}(f(x), y)$ for all $x, y$.
  \item \textbf{E-reducibility (error-preserving)} \cite{Khanna1998}: $R_{P_1}(x, g(x,y)) - 1 \le c \bigl(R_{P_2}(f(x), y) - 1\bigr)$ for a constant $c > 0$. Every E-reduction preserves membership in $\FPTAS$, $\PTAS$, and $\APX$.
  \item \textbf{P-reducibility / PTAS-reducibility} \cite{Crescenzi1999}: for every $\varepsilon > 0$ there exists $\delta(\varepsilon) > 0$ such that $R_{P_2}(f(x, \varepsilon), y) \le 1 + \delta(\varepsilon)$ implies $R_{P_1}(x, g(x, y, \varepsilon)) \le 1 + \varepsilon$.
  \item \textbf{A-reducibility and Continuous reducibility} \cite{Crescenzi1999}: for every $r > 1$ there exists $r' > 1$ such that an $r'$-approximation for $P_2$ yields an $r$-approximation for $P_1$ (hence preserving $\APX$).
\end{itemize}
Figure~\ref{fig:reduction-taxonomy} illustrates the implication relationships among these approximation-preserving reducibilities \cite{Ausiello1999,Crescenzi1999}. Within $\APX$, any L-reduction from a maximization problem $P_1 \in \APX$ whose optimum is within a constant factor of $\OPT_{P_2}(f(x))$ induces both an E-reduction and a PTAS-reduction. We say that a problem $P \in \APX$ is \emph{$\APX$-complete} if every problem in $\APX$ reduces to $P$ under PTAS-reductions (or L-reductions). By Lemma~\ref{lem:l-red-ptas}, if any $\APX$-complete problem admits a $\PTAS$, then $\APX = \PTAS$.
\end{remark}

\begin{figure}[t]
\centering
\begin{tikzpicture}[
  node distance=9.5mm and 14mm,
  box/.style={draw=black!85, rounded corners=1.5pt, line width=0.55pt, fill=white, inner sep=5.5pt, minimum width=38mm, font=\small},
  arr/.style={-{Stealth[length=2.2mm,width=1.6mm]}, line width=0.6pt, draw=black!85}
]
  \node[box] (ptas) {PTAS-reducibility};
  \node[box, below=9.5mm of ptas, xshift=-26mm] (pred) {P-reducibility (AP-reducibility)};
  \node[box, below=9.5mm of ptas, xshift=26mm] (ared) {A-reducibility};
  \node[box, below=9.5mm of pred, xshift=-22mm] (lred) {L-reducibility};
  \node[box, below=9.5mm of ptas, yshift=-19mm] (ered) {E-reducibility};
  \node[box, below=9.5mm of ared, xshift=22mm] (cred) {Continuous reducibility};
  \node[box, below=9.5mm of ered] (sred) {Strict reducibility};

  \draw[arr] (pred) -- (ptas);
  \draw[arr] (lred) -- (pred);
  \draw[arr] (ered) -- (pred);
  \draw[arr] (ered) -- (ared);
  \draw[arr] (cred) -- (ared);
  \draw[arr] (sred) -- (ered);
\end{tikzpicture}
\caption{Implication hierarchy among approximation-preserving reducibilities \cite{Papadimitriou1991,Khanna1998,Crescenzi1999,Ausiello1999}. An arrow $\mathcal{R}_1 \to \mathcal{R}_2$ indicates that every $\mathcal{R}_1$-reduction (restricted to $\NPO$ problems with linearly bounded optima in the case of L-reducibility) is also an $\mathcal{R}_2$-reduction.}
\label{fig:reduction-taxonomy}
\end{figure}

\section{Completeness in APX and Explicit L-Reductions}
\label{sec:completeness}

We now use \prob{Max-3SAT} as the starting point for establishing $\APX$-completeness of core combinatorial optimization problems via explicit L-reductions. The bridge from $\NP$-completeness to $\APX$-hardness is provided by the PCP theorem \cite{Arora1998PCP,Arora1998Safra,Dinur2007}.

\begin{lemma}[PCP Theorem, Gap-\prob{3SAT} formulation \cite{Arora1998PCP,Dinur2007}]
\label{lem:pcp}
There exists a constant $\rho \in (0, 1)$ such that for every language $L \in \NP$, there is a polynomial-time computable function $f$ mapping strings $x \in \{0,1\}^*$ to $3$-CNF formulas $f(x)$ satisfying:
\begin{align}
  x \in L &\implies \overline{\OPT}(f(x)) = 1, \label{eq:pcp-yes}\\
  x \notin L &\implies \overline{\OPT}(f(x)) < \rho, \label{eq:pcp-no}
\end{align}
where $\overline{\OPT}(f(x)) = \OPT(f(x))/m$ denotes the maximum fraction of clauses of $f(x)$ satisfiable by any truth assignment.
\end{lemma}

By H\aa stad's optimal inapproximability theorem \cite{Hastad2001}, Lemma~\ref{lem:pcp} holds with $\rho = 7/8 + \eta$ for any constant $\eta > 0$. Combining Lemma~\ref{lem:pcp} with the closure of $\MaxSNP$ under PTAS/L-reductions yields the following foundational theorem \cite{Papadimitriou1991,Khanna1998,AroraBarak2009}:

\begin{theorem}[\cite{Papadimitriou1991,Arora1998PCP}]
\label{thm:max3sat-apx-hard}
\prob{Max-3SAT} is $\APX$-complete.
\end{theorem}

We provide the proof sketch reducing \prob{Max-$q$CSP} to \prob{Max-3SAT} in Section~\ref{sec:ph} (see also Arora and Barak \cite[Section 18.5]{AroraBarak2009}). Here we record its immediate consequence for all $\APX$-complete problems.

\begin{corollary}
\label{cor:no-ptas}
If $\Pclass \neq \NP$, then no $\APX$-complete problem admits a $\PTAS$.
\end{corollary}
\begin{proof}
Suppose $\Pclass \neq \NP$, and assume for contradiction that some $\APX$-complete problem admits a $\PTAS$. By Lemma~\ref{lem:l-red-ptas} and Theorem~\ref{thm:max3sat-apx-hard}, \prob{Max-3SAT} also admits a $\PTAS$. Let $L \in \NP$ be an arbitrary language and let $\rho < 1$ and $f$ be given by Lemma~\ref{lem:pcp}. Set $\varepsilon = \frac{1 - \rho}{2} \in (0, 1)$, so that $1 - \varepsilon = \frac{1 + \rho}{2} > \rho$. Running the $(1 - \varepsilon)$-approximation algorithm $\mathcal{A}_\varepsilon$ for \prob{Max-3SAT} on $f(x)$ returns a truth assignment whose satisfied fraction $\overline{\ALG}(f(x))$ obeys:
\begin{itemize}[topsep=2pt,itemsep=2pt]
  \item if $x \in L$, then $\overline{\OPT}(f(x)) = 1$, so $\overline{\ALG}(f(x)) \ge 1 - \varepsilon$;
  \item if $x \notin L$, then $\overline{\OPT}(f(x)) < \rho$, so $\overline{\ALG}(f(x)) \le \overline{\OPT}(f(x)) < \rho < 1 - \varepsilon$.
\end{itemize}
Thus $x \in L$ if and only if $\overline{\ALG}(f(x)) \ge 1 - \varepsilon$, deciding $L$ in polynomial time and contradicting $\Pclass \neq \NP$.
\end{proof}

We now show that \prob{Max-2SAT} is also $\APX$-complete. This is a striking contrast with the decision problem \prob{2SAT}, which is solvable in linear time $O(n + m)$ by strongly connected components on the implication graph.

\begin{theorem}[Garey, Johnson, and Stockmeyer \cite{Garey1976}; Papadimitriou and Yannakakis \cite{Papadimitriou1991}]
\label{thm:max2sat}
\prob{Max-2SAT} is $\APX$-complete under L-reductions.
\end{theorem}
\begin{proof}
Membership in $\APX$ follows from Proposition~\ref{prop:max3sat-in-apx} since every $2$-CNF formula is a $3$-CNF formula. To prove $\APX$-hardness, we construct an L-reduction from \prob{Max-3SAT} to \prob{Max-2SAT}:
\[
  \prob{Max-3SAT} \le_L \prob{Max-2SAT}.
\]
Let $\varphi$ be an instance of \prob{Max-3SAT} with $m = |C|$ clauses $C_1, \dots, C_m$. For each $3$-CNF clause $C_j = (x \lor y \lor z)$, introduce a fresh auxiliary variable $v_j$ and replace $C_j$ by the collection $f(C_j)$ of ten clauses of width at most $2$ \cite{Garey1976}:
\begin{equation}
  (x),\; (y),\; (z),\; (v_j),\;
  (\neg x \lor \neg y),\; (\neg x \lor \neg z),\; (\neg y \lor \neg z),\;
  (x \lor \neg v_j),\; (y \lor \neg v_j),\; (z \lor \neg v_j).
  \label{eq:gjs-gadget}
\end{equation}
Let $f(\varphi)$ be the $2$-CNF formula formed by taking the conjunction of these $10m$ clauses across all $j \in \{1, \dots, m\}$. By symmetry of the ten clauses in $(x, y, z)$, the number of satisfied clauses in $f(C_j)$ depends only on the number $s \in \{0, 1, 2, 3\}$ of true literals among $\{x, y, z\}$ and on the truth value of $v_j \in \{0, 1\}$, as tabulated in Table~\ref{tab:gjs-gadget}.

\begin{table}[t]
\centering
\caption{Number of satisfied clauses in the 10-clause \prob{Max-2SAT} gadget~\eqref{eq:gjs-gadget} as a function of the number of true literals $s = x + y + z \in \{0,1,2,3\}$ in $C_j = (x \lor y \lor z)$ and the auxiliary variable $v_j \in \{0,1\}$.}
\label{tab:gjs-gadget}
\small
\setlength{\tabcolsep}{4.5pt}
\begin{tabular}{@{}lccccc@{}}
\toprule
True literals $s$ & Unit $(x),(y),(z)$ & $(\neg x \lor \neg y)$ etc. & $v_j = 0$ & $v_j = 1$ & $\max_{v_j}$ \\
\midrule
$s = 0$ ($C_j$ false) & $0$ & $3$ & $0 + 3 + 0 + 3 = 6$ & $0 + 3 + 1 + 0 = 4$ & $\mathbf{6}$ \\
$s = 1$ ($C_j$ true)  & $1$ & $3$ & $1 + 3 + 0 + 3 = 7$ & $1 + 3 + 1 + 1 = 6$ & $\mathbf{7}$ \\
$s = 2$ ($C_j$ true)  & $2$ & $2$ & $2 + 2 + 0 + 3 = 7$ & $2 + 2 + 1 + 2 = 7$ & $\mathbf{7}$ \\
$s = 3$ ($C_j$ true)  & $3$ & $0$ & $3 + 0 + 0 + 3 = 6$ & $3 + 0 + 1 + 3 = 7$ & $\mathbf{7}$ \\
\bottomrule
\end{tabular}
\end{table}

Table~\ref{tab:gjs-gadget} shows that if $C_j$ is satisfied ($s \ge 1$), choosing $v_j$ optimally satisfies exactly $7$ of the $10$ clauses in $f(C_j)$ and no choice of $v_j$ satisfies more than $7$. Conversely, if $C_j$ is falsified ($s = 0$), at most $6$ clauses in $f(C_j)$ can be satisfied (attained at $v_j = 0$).

Given any truth assignment $y$ for $f(\varphi)$, define $g(y)$ to be the restriction of $y$ to the original variables of $\varphi$. If $g(y)$ satisfies $k = \VAL_{\prob{Max-3SAT}}(g(y))$ clauses of $\varphi$, then in $f(\varphi)$ the assignment $y$ satisfies at most $7$ clauses in each of the $k$ satisfied gadgets and at most $6$ clauses in each of the $m - k$ falsified gadgets, so
\begin{equation}
  \VAL_{\prob{Max-2SAT}}(y) \le 7k + 6(m - k) = 6m + \VAL_{\prob{Max-3SAT}}(g(y)),
  \label{eq:max2sat-val-bound}
\end{equation}
with equality whenever $y$ chooses each $v_j$ optimally for $g(y)$. In particular, maximizing over all assignments gives
\begin{equation}
  \OPT_{\prob{Max-2SAT}}(f(\varphi)) = 6m + \OPT_{\prob{Max-3SAT}}(\varphi).
  \label{eq:max2sat-opt}
\end{equation}
By Proposition~\ref{prop:max3sat-in-apx}, $\OPT_{\prob{Max-3SAT}}(\varphi) \ge m/2$, so $6m \le 12\,\OPT_{\prob{Max-3SAT}}(\varphi)$ and therefore
\[
  \OPT_{\prob{Max-2SAT}}(f(\varphi)) = 6m + \OPT_{\prob{Max-3SAT}}(\varphi) \le 13\,\OPT_{\prob{Max-3SAT}}(\varphi),
\]
establishing~\eqref{eq:l-red-opt} with $a = 13$. Subtracting~\eqref{eq:max2sat-val-bound} from~\eqref{eq:max2sat-opt} gives, for every assignment $y$ to $f(\varphi)$,
\[
  \OPT_{\prob{Max-3SAT}}(\varphi) - \VAL_{\prob{Max-3SAT}}(g(y))
  \le \OPT_{\prob{Max-2SAT}}(f(\varphi)) - \VAL_{\prob{Max-2SAT}}(y),
\]
which satisfies~\eqref{eq:l-red-err} with $b = 1$.
\end{proof}

Next, let \prob{Max-NAE-3SAT} (Maximum Not-All-Equal 3SAT) be the problem in which we are given a collection of clauses of size at most $3$ and seek a truth assignment maximizing the number of clauses that contain at least one true literal and at least one false literal. The following two reductions are due to Papadimitriou and Yannakakis \cite{Papadimitriou1991}.

\begin{theorem}[Papadimitriou and Yannakakis \cite{Papadimitriou1991}]
\label{thm:max-nae-3sat}
\prob{Max-NAE-3SAT} is $\APX$-complete under L-reductions.
\end{theorem}
\begin{proof}
We construct an L-reduction $\prob{Max-2SAT} \le_L \prob{Max-NAE-3SAT}$. Let $\varphi$ be an instance of \prob{Max-2SAT} with clauses $C_1, \dots, C_m$, where each clause $C_j = (\ell_{j,1} \lor \ell_{j,2})$ has one or two literals. Introduce a single fresh variable $z_0$ shared across all clauses, and define $f(\varphi)$ to be the \prob{Max-NAE-3SAT} instance with clauses
\[
  C_j' = \operatorname{NAE}(\ell_{j,1}, \ell_{j,2}, z_0), \qquad j = 1, \dots, m.
\]
Notice that a not-all-equal clause $\operatorname{NAE}(\ell_{j,1}, \ell_{j,2}, z_0)$ is invariant under complementing all variables simultaneously: an assignment $y$ satisfies $\operatorname{NAE}(\ell_{j,1}, \ell_{j,2}, z_0)$ if and only if its bitwise negation $\overline{y}$ does. Given any truth assignment $y$ to the variables of $f(\varphi)$, define $g(y)$ on the variables of $\varphi$ to be the restriction of $y$ if $y(z_0) = 0$, and the restriction of $\overline{y}$ if $y(z_0) = 1$. In either case, under the normalized assignment the third literal $z_0$ evaluates to $0$ (false), so $\operatorname{NAE}(\ell_{j,1}, \ell_{j,2}, 0)$ is satisfied if and only if at least one of $\ell_{j,1}$ or $\ell_{j,2}$ evaluates to $1$ (true), which is equivalent to $g(y)$ satisfying the $2$-CNF clause $(\ell_{j,1} \lor \ell_{j,2})$. Consequently, $\VAL_{\prob{Max-2SAT}}(g(y)) = \VAL_{\prob{Max-NAE-3SAT}}(y)$ for every $y$, and $\OPT_{\prob{Max-NAE-3SAT}}(f(\varphi)) = \OPT_{\prob{Max-2SAT}}(\varphi)$. This is an L-reduction with $a = 1$ and $b = 1$.
\end{proof}

For an undirected graph (or multigraph) $G = (V, E)$, \prob{Max-Cut} asks for a partition $(S, T)$ of $V$ maximizing the number of edges $E(S, T)$ with one endpoint in $S$ and one endpoint in $T$.

\begin{theorem}[Papadimitriou and Yannakakis \cite{Papadimitriou1991}]
\label{thm:max-cut}
\prob{Max-Cut} is $\APX$-complete under L-reductions.
\end{theorem}
\begin{proof}
Membership in $\APX$ follows from the $1/2$-approximation (Section~\ref{sec:algorithms}). To establish $\APX$-hardness, we construct an L-reduction $\prob{Max-NAE-3SAT} \le_L \prob{Max-Cut}$. We present the reduction for multigraphs first and then explain the simple-graph gadget.

Let $x$ be an instance of \prob{Max-NAE-3SAT} with $m$ clauses over variables $v_1, \dots, v_n$, and let $k_i$ denote the total number of occurrences of $v_i$ and $\neg v_i$ across all clauses, so that $\ell = \sum_{i=1}^n k_i \le 3m$ is the total number of literal occurrences. Construct the multigraph $G = f(x) = (V, E)$ as follows:
\begin{enumerate}[label=(\roman*),topsep=2pt,itemsep=2pt]
  \item For each variable $v_i$, create two vertices $v_i$ and $\neg v_i$ in $V$.
  \item For each 3-literal clause $(a \lor b \lor c)$, add a triangle of edges $(a,b)$, $(b,c)$, and $(a,c)$ to $E$. For a 2-literal clause $(a \lor b)$, add two parallel edges between $a$ and $b$. (Thus each clause contributes $3$ or $2$ clause edges, and is NAE-satisfied by a truth assignment if and only if exactly $2$ of its clause edges cross the cut between true and false literals; if all literals in the clause have the same truth value, $0$ of its clause edges cross the cut.)
  \item For each variable $v_i$, add $2k_i$ parallel consistency edges between $v_i$ and $\neg v_i$.
\end{enumerate}
Figure~\ref{fig:maxcut-reduction} illustrates this construction on the two-clause formula $(a \lor b \lor c) \land (\neg a \lor b \lor d)$, where $k_a = k_b = 2$ (four consistency edges) and $k_c = k_d = 1$ (two consistency edges).

\begin{figure}[t]
\centering
\begin{tikzpicture}[
  lit/.style={draw=black!85, circle, line width=0.6pt, minimum size=8.2mm, inner sep=1pt, font=\small, fill=white},
  clauseedge/.style={line width=0.8pt, draw=black!90},
  consedge/.style={line width=0.55pt, densely dashed, draw=black!65}
]
  \node[lit] (a)  at (0.0, 2.2) {$a$};
  \node[lit] (na) at (0.0, 0.0) {$\neg a$};
  \node[lit] (b)  at (3.0, 2.2) {$b$};
  \node[lit] (nb) at (5.8, 2.2) {$\neg b$};
  \node[lit] (c)  at (1.5, 4.1) {$c$};
  \node[lit] (nc) at (4.3, 4.1) {$\neg c$};
  \node[lit] (d)  at (3.0, 0.0) {$d$};
  \node[lit] (nd) at (5.8, 0.0) {$\neg d$};

  \draw[clauseedge] (a) -- (b);
  \draw[clauseedge] (b) -- (c);
  \draw[clauseedge] (c) -- (a);

  \draw[clauseedge] (na) -- (b);
  \draw[clauseedge] (b) -- (d);
  \draw[clauseedge] (d) -- (na);

  \draw[consedge] (a) edge[bend right=24] (na);
  \draw[consedge] (a) edge[bend right=8]  (na);
  \draw[consedge] (a) edge[bend left=8]   (na);
  \draw[consedge] (a) edge[bend left=24]  (na);

  \draw[consedge] (b) edge[bend left=24]  (nb);
  \draw[consedge] (b) edge[bend left=8]   (nb);
  \draw[consedge] (b) edge[bend right=8]  (nb);
  \draw[consedge] (b) edge[bend right=24]  (nb);

  \draw[consedge] (c) edge[bend left=14]  (nc);
  \draw[consedge] (c) edge[bend right=14] (nc);

  \draw[consedge] (d) edge[bend left=14]  (nd);
  \draw[consedge] (d) edge[bend right=14] (nd);
\end{tikzpicture}
\caption{Multigraph $G = f(x)$ constructed in Theorem~\ref{thm:max-cut} for the \prob{Max-NAE-3SAT} formula $x = (a \lor b \lor c) \land (\neg a \lor b \lor d)$. Solid edges form the clause triangles for $(a \lor b \lor c)$ and $(\neg a \lor b \lor d)$; dashed parallel bundles place $2k_i$ consistency edges between each literal pair $(v_i, \neg v_i)$, with $2k_a = 2k_b = 4$ and $2k_c = 2k_d = 2$.}
\label{fig:maxcut-reduction}
\end{figure}

First, let $\tau$ be any truth assignment satisfying $p$ clauses of $x$. Partition $V$ into $(S, T)$ by placing all true literals in $S$ and all false literals in $T$. For each variable $v_i$, $v_i$ and $\neg v_i$ lie on opposite sides of the cut, contributing $\sum_{i=1}^n 2k_i = 2\ell$ consistency edges. Each of the $p$ satisfied NAE clauses contributes $2$ clause edges across $(S, T)$, and each falsified NAE clause has all its literals on the same side and contributes $0$ clause edges. Thus this cut has value $2\ell + 2p$.

Conversely, let $y = (S, T)$ be an arbitrary cut of $G$. We define the polynomial-time decoding function $g(y)$ by first converting $(S, T)$ into a \emph{consistent} cut $(S', T')$ where $v_i$ and $\neg v_i$ lie on opposite sides for every $i$: if both $v_i$ and $\neg v_i$ belong to the same side (say $S$), moving $v_i$ to $T$ gains all $2k_i$ consistency edges between $v_i$ and $\neg v_i$ while losing at most $2k_i$ clause edges incident to $v_i$ (since each occurrence of $v_i$ in a clause has at most $2$ incident clause edges). Thus moving $v_i$ across the cut never decreases the cut size:
\[
  \VAL_{\prob{Max-Cut}}(S', T') \ge \VAL_{\prob{Max-Cut}}(S, T).
\]
Because $(S', T')$ separates every pair $(v_i, \neg v_i)$, setting literals in $S'$ to true and literals in $T'$ to false defines a valid truth assignment $g(y)$ for $x$. If $g(y)$ satisfies $p = \VAL_{\prob{Max-NAE-3SAT}}(g(y))$ clauses of $x$, then $\VAL_{\prob{Max-Cut}}(S', T') = 2\ell + 2p$. In particular,
\[
  \OPT_{\prob{Max-Cut}}(f(x)) = 2\ell + 2\OPT_{\prob{Max-NAE-3SAT}}(x).
\]
Since $\ell \le 3m$ and a random assignment satisfies each NAE clause of length $2$ or $3$ with probability at least $1/2$, we have $\OPT_{\prob{Max-NAE-3SAT}}(x) \ge m/2$, so $m \le 2\OPT_{\prob{Max-NAE-3SAT}}(x)$. Therefore,
\begin{align*}
  \OPT_{\prob{Max-Cut}}(f(x))
  &= 2\ell + 2\OPT_{\prob{Max-NAE-3SAT}}(x) \\
  &\le 6m + 2\OPT_{\prob{Max-NAE-3SAT}}(x)
  \le 14\OPT_{\prob{Max-NAE-3SAT}}(x),
\end{align*}
satisfying~\eqref{eq:l-red-opt} with $a = 14$. Furthermore, for any cut $y = (S, T)$,
\begin{align*}
  &\OPT_{\prob{Max-NAE-3SAT}}(x) - \VAL_{\prob{Max-NAE-3SAT}}(g(y)) \\
  &\qquad= \frac{1}{2}\bigl(\OPT_{\prob{Max-Cut}}(f(x)) - \VAL_{\prob{Max-Cut}}(S', T')\bigr) \\
  &\qquad\le \frac{1}{2}\bigl(\OPT_{\prob{Max-Cut}}(f(x)) - \VAL_{\prob{Max-Cut}}(y)\bigr),
\end{align*}
satisfying~\eqref{eq:l-red-err} with $b = 1/2$.
\end{proof}

\begin{remark}[Simple graphs and bounded-degree graph problems]
To turn the multigraph $G$ in Theorem~\ref{thm:max-cut} into a simple graph without parallel edges, one replaces each bundle of $2k_i$ parallel edges between $v_i$ and $\neg v_i$ by a bipartite gadget (or constant-degree expander gadget \cite{Papadimitriou1991}) through intermediate auxiliary vertices so that separating $v_i$ from $\neg v_i$ remains strictly optimal. Starting from bounded-occurrence variants of \prob{Max-3SAT}, Alimonti and Kann \cite{Alimonti2000} further proved that \prob{Min-Vertex-Cover-$B$} and \prob{Max-Independent-Set-$B$} on graphs of maximum degree $B \ge 3$ (cubic graphs) are $\APX$-complete under L-reductions.
\end{remark}

\section{The Approximability Hierarchy, Self-Improvability, and the Polynomial Hierarchy}
\label{sec:ph}

\subsection{Gap Reductions and the PCP Theorem}

While classical complexity classes are defined over decision problems, hardness of approximation for search and optimization problems in $\NPO$ is established by reducing $\NP$ decision problems to \emph{promise gap problems} \cite{AroraBarak2009}.

\begin{proposition}[Gap problems and hardness of approximation]
\label{prop:gap-to-inapprox}
Let $\Pi \in \NPO$ be a maximization (respectively, minimization) problem and let $c(n) > s(n) > 0$. Suppose it is $\NP$-hard to distinguish instances $I \in \mathcal{I}_\Pi$ with $\OPT_\Pi(I) \ge c(|I|)$ from instances with $\OPT_\Pi(I) \le s(|I|)$. Then no polynomial-time algorithm can approximate $\Pi$ within any factor strictly smaller than $c(|I|)/s(|I|)$ unless $\Pclass = \NP$.
\end{proposition}
\begin{proof}
If a polynomial-time algorithm $\mathcal{A}$ satisfies $\OPT_\Pi(I)/\ALG(I) < c(|I|)/s(|I|)$ on a maximization problem, then on any instance with $\OPT_\Pi(I) \ge c(|I|)$ it outputs $\ALG(I) > s(|I|)$, whereas on any instance with $\OPT_\Pi(I) \le s(|I|)$ every feasible solution has value at most $s(|I|)$. Comparing $\ALG(I)$ against the threshold $s(|I|)$ therefore solves the $\NP$-hard gap problem in polynomial time. The minimization case is identical with the inequalities reversed.
\end{proof}

Under the Cook--Levin theorem, a certificate for an $\NP$ computation is a satisfying assignment to a Boolean formula encoding a Turing machine tableau; flipping even a single bit of the witness can change the formula from satisfiable to unsatisfiable, so a standard verifier must inspect the entire witness. The PCP theorem ($\NP = \mathsf{PCP}[O(\log n), O(1)]$) shows that any $\NP$ witness can be encoded into a proof whose validity can be checked by a randomized polynomial-time verifier that reads only a constant number $q = O(1)$ of proof bits using $O(\log n)$ random bits \cite{Arora1998PCP,Arora1998Safra,Dinur2007}. Enumerating all $2^{O(\log n)} = \poly(n)$ random strings of the verifier produces an instance of \prob{Max-$q$CSP} with $m = \poly(n)$ constraints over $q$ variables each, where a valid proof satisfies all $m$ constraints (completeness $1$) and any proof for a non-member $x \notin L$ violates at least half of the constraints (soundness $1/2$). From this gap-\prob{Max-$q$CSP} instance, hardness of approximating \prob{Max-3SAT} follows by a local clause expansion:

\begin{proposition}[From Gap-\prob{Max-$q$CSP} to Gap-\prob{Max-3SAT} \cite{Dinur2007,AroraBarak2009}]
\label{prop:qcsp-to-3sat}
Suppose there is a constant $q \ge 3$ such that distinguishing satisfiable instances of \prob{Max-$q$CSP} ($\OPT(I) = m$) from instances with $\OPT(I) \le m/2$ is $\NP$-hard. Then for $\varepsilon_1 = \frac{1}{2q 2^q} > 0$, there is no polynomial-time $(1 - \varepsilon_1)$-approximation algorithm for \prob{Max-3SAT} unless $\Pclass = \NP$.
\end{proposition}
\begin{proof}
Let $I$ be an instance of \prob{Max-$q$CSP} with $m$ Boolean constraints $f_1, \dots, f_m$, where each constraint $f_j(x_{j_1}, \dots, x_{j_q}) \in \{0, 1\}$ depends on at most $q$ variables. Any Boolean function on $q$ variables can be expressed as a $q$-CNF formula with at most $2^q$ clauses (one clause ruling out each falsifying assignment in $\{0,1\}^q$). Next, each $q$-clause $(\ell_1 \lor \ell_2 \lor \dots \lor \ell_q)$ is converted into an equisatisfiable conjunction of at most $q$ clauses of width $3$ by introducing $q-3$ fresh auxiliary variables $z_1, \dots, z_{q-3}$:
\[
  (\ell_1 \lor \ell_2 \lor z_1) \land (\neg z_1 \lor \ell_3 \lor z_2) \land \dots \land (\neg z_{q-3} \lor \ell_{q-1} \lor \ell_q).
\]
Thus each constraint $f_j$ is replaced by a $3$-CNF subformula $\phi_j$ containing $m_j \le q 2^q$ clauses, and the combined $3$-CNF formula $\phi_I = \bigwedge_{j=1}^m \phi_j$ has $m' = \sum_{j=1}^m m_j \le q 2^q m$ clauses.

If $I$ is satisfiable, setting the original variables $x_1, \dots, x_n$ according to a satisfying assignment of $I$ and setting the auxiliary variables $z$ appropriately satisfies every clause of $\phi_I$, so $\OPT_{\prob{Max-3SAT}}(\phi_I) = m'$. Conversely, if $\OPT_{\prob{Max-}q\prob{CSP}}(I) \le m/2$, then for any truth assignment to the variables of $\phi_I$, its restriction to $x_1, \dots, x_n$ violates at least $m/2$ constraints $f_j$ of $I$. For each violated constraint $f_j$, at least one $3$-clause in $\phi_j$ must be falsified. Hence at least $m/2$ clauses of $\phi_I$ are falsified, so
\[
  \OPT_{\prob{Max-3SAT}}(\phi_I) \le m' - \frac{m}{2}
  = m'\!\left(1 - \frac{m}{2m'}\right)
  \le m'\!\left(1 - \frac{1}{2q 2^q}\right)
  = (1 - \varepsilon_1) m'.
\]
By Proposition~\ref{prop:gap-to-inapprox}, approximating \prob{Max-3SAT} within factor $1/(1 - \varepsilon_1)$ is $\NP$-hard.
\end{proof}

\subsection{Hierarchy of Approximation Classes and Self-Improvability}

Within $\NPO$, optimization problems stratify into a five-level hierarchy of increasing approximation ratios:
\begin{equation}
  \FPTAS \;\subsetneq\; \PTAS \;\subsetneq\; \APX \;\subsetneq\; \logAPX \;\subsetneq\; \polyAPX \;\subseteq\; \NPO,
  \label{eq:apx-hierarchy}
\end{equation}
where each strict inclusion holds if and only if $\Pclass \neq \NP$. Table~\ref{tab:hierarchy} summarizes representative problems at each level alongside their best polynomial-time approximation factors and matching hardness bounds.

\begin{table}[t]
\centering
\caption{Stratification of $\NPO$ optimization problems across the approximability hierarchy~\eqref{eq:apx-hierarchy}, showing representative problems, polynomial-time algorithmic bounds, and conditional inapproximability thresholds.}
\label{tab:hierarchy}
\footnotesize
\setlength{\tabcolsep}{2.5pt}
\begin{tabular}{@{}lllll@{}}
\toprule
Class & Ratio $R_\Pi(I)$ & Problem & Algorithmic Bound & Hardness ($\Pclass \neq \NP$) \\
\midrule
$\FPTAS$ & $1 + \varepsilon$ ($\poly(n,1/\varepsilon)$) & \prob{0-1 Knapsack} & $1 + \varepsilon$ in $O(n^3/\varepsilon)$ & No strongly $\NP$-hard \\
$\PTAS$  & $1 + \varepsilon$ ($n^{O(1/\varepsilon)}$) & \prob{Euclidean TSP} & $1 + \varepsilon$ \cite{Arora1998TSP} & No $\FPTAS$ \\
$\APX$   & $O(1)$ & \prob{Max-3SAT} & $8/7$ \cite{Johnson1974} & $8/7 - \varepsilon$ \cite{Hastad2001} \\
         &        & \prob{Max-Cut} & $1.1383$ \cite{Goemans1995} & $17/16 - \varepsilon$ ($1.1383$ UGC \cite{Khot2007}) \\
         &        & \prob{Min-Vertex-Cover} & $2$ & $1.3606$ \cite{DinurSafra2005} ($2-\varepsilon$ UGC \cite{KhotRegev2008}) \\
         &        & \prob{Metric-TSP} & $3/2 - \varepsilon_0$ \cite{Christofides1976,Karlin2021} & $123/122 - \varepsilon$ \cite{Karpinski2015} \\
$\logAPX$ & $O(\log n)$ & \prob{Min-Set-Cover} & $\ln n + 1$ \cite{Johnson1974,Vazirani2001} & $(1 - o(1))\ln n$ \cite{DinurSteurer2014} \\
$\polyAPX$ & $O(n^c)$ & \prob{Max-Clique}, \prob{Max-IS} & $O(n / \log^2 n)$ & $n^{1 - \varepsilon}$ \cite{Hastad1999,Zuckerman2007} \\
\bottomrule
\end{tabular}
\end{table}

Beyond separating classes under $\Pclass \neq \NP$, Crescenzi, Kann, Silvestri, and Trevisan \cite{Crescenzi1999} connected the internal structure of $\APX$, $\logAPX$, and $\polyAPX$ to the Polynomial Hierarchy ($\mathsf{PH}$) and to the algebraic property of \emph{self-improvability}.

\begin{definition}[Self-improvability \cite{Crescenzi1999,Ausiello1999}]
\label{def:self-improvable}
An optimization problem $\Pi \in \NPO$ is \emph{self-improvable} if there exist two polynomial-time algorithms $\mathcal{A}_1$ and $\mathcal{A}_2$ such that for every instance $I \in \mathcal{I}_\Pi$ and any two rational approximation factors $r_1, r_2 > 1$:
\begin{enumerate}[label=(\roman*),topsep=2pt,itemsep=2pt]
  \item $I' = \mathcal{A}_1(I, r_1, r_2) \in \mathcal{I}_\Pi$ is an instance of $\Pi$; and
  \item for every feasible solution $y' \in \mathrm{Sol}_\Pi(I')$ achieving performance ratio $R_\Pi(I', y') \le r_2$, the reconstructed solution $y = \mathcal{A}_2(I, y', r_1, r_2) \in \mathrm{Sol}_\Pi(I)$ achieves performance ratio $R_\Pi(I, y) \le r_1$.
\end{enumerate}
\end{definition}

When a problem is self-improvable (for instance, \prob{Max-Clique} via the graph lexicographic product $G \times G$, which squares the clique number $\omega(G \times G) = \omega(G)^2$ and converts any constant-factor $r_2$-approximation on $G^k$ into an $r_2^{1/k}$-approximation on $G$), any single constant-factor approximation algorithm can be boosted to a $\PTAS$. This observation yields three structural consequences \cite{Crescenzi1999}:

\begin{proposition}[Structural consequences of self-improvability and $\mathsf{PH}$ \cite{Crescenzi1999}]
\label{prop:struct-apx}
\leavevmode
\begin{enumerate}[label=(\arabic*),topsep=2pt,itemsep=2pt]
  \item \textbf{Non-self-improvability of $\APX$-hard problems:} No $\APX$-complete, $\logAPX$-complete, or $\polyAPX$-complete problem is self-improvable unless $\Pclass = \NP$ (or $\NP \subseteq \mathsf{DTIME}(n^{\poly(\log n)})$ under quasipolynomial self-improvement).
  \item \textbf{Incompleteness of \prob{Min-Bin-Packing}:} The problem \prob{Min-Bin-Packing} is not $\APX$-complete under PTAS-reductions unless the Polynomial Hierarchy collapses to $\Sigma_2^p = \Pi_2^p$.
  \item \textbf{Booster property of \prob{Max-Clique}:} Because \prob{Max-Clique} is self-improvable via graph products, approximating \prob{Max-Clique} within factor $2$ is equivalent in polynomial time to obtaining a $\PTAS$ for \prob{Max-Clique} (and hence impossible unless $\Pclass = \NP$).
\end{enumerate}
\end{proposition}

\section{APX-Hardness in Bayesian Mechanism Design}
\label{sec:mechanism-design}

Approximation complexity also governs economic allocation problems. A central question in microeconomic theory and algorithmic game theory is \emph{Bayesian optimal mechanism design}: a seller wishes to allocate a single indivisible item among $n$ strategic bidders whose private valuations $\mathbf{v} = (v_1, \dots, v_n) \in V^n \subseteq [0, 1]^n$ are drawn from a known joint prior distribution $\varphi(\mathbf{v})$. When the bidders' valuations are statistically independent ($\varphi(\mathbf{v}) = \prod_{i=1}^n \varphi_i(v_i)$), Myerson's lemma \cite{Myerson1981} gives an exact, closed-form optimal auction via virtual valuation ironings. When valuations are \emph{correlated}, however, Papadimitriou and Pierrakos \cite{Papadimitriou2011} (see also Daskalakis, Deckelbaum, and Tzamos \cite{Daskalakis2014} for multi-item auctions) showed that finding the revenue-maximizing deterministic auction becomes $\APX$-hard as soon as there are $n = 3$ bidders.

\subsection{Formulation of Optimal Single-Item Auctions}

Consider $n$ bidders with finite valuation support $V \subset [0, 1]$ and a joint probability mass function $\varphi: V^n \to \mathbb{Q}_{\ge 0}$. A deterministic direct-revelation auction specifies an allocation rule $x_i: V^n \to \{0, 1\}$ with $\sum_{i=1}^n x_i(\mathbf{v}) \le 1$ (indicating which bidder, if any, receives the item) and a payment rule $p_i: V^n \to \mathbb{R}_{\ge 0}$. Writing $\mathbf{v} = (v_i, \mathbf{v}_{-i})$ where $\mathbf{v}_{-i}$ denotes the valuations of all bidders other than $i$, the auctioneer seeks to maximize expected revenue
\begin{equation}
  \operatorname{Rev}(x, p) = \sum_{\mathbf{v} \in V^n} \varphi(\mathbf{v}) \sum_{i=1}^n p_i(\mathbf{v})
  \label{eq:auction-revenue}
\end{equation}
subject to two economic constraints for all bidders $i \in \{1, \dots, n\}$, all true valuations $v_i \in V$, all misreports $v_i' \in V$, and all opponent profiles $\mathbf{v}_{-i} \in V^{n-1}$:
\begin{enumerate}[label=(\arabic*),topsep=2pt,itemsep=2pt]
  \item \textbf{Dominant-Strategy Incentive Compatibility (DSIC):}
  \begin{equation}
    v_i x_i(v_i, \mathbf{v}_{-i}) - p_i(v_i, \mathbf{v}_{-i}) \;\ge\; v_i x_i(v_i', \mathbf{v}_{-i}) - p_i(v_i', \mathbf{v}_{-i}).
    \label{eq:dsic}
  \end{equation}
  Bidder $i$ cannot strictly increase utility by reporting $v_i' \neq v_i$, regardless of the other bidders' reports.
  \item \textbf{Ex-Post Individual Rationality (IR):}
  \begin{equation}
    v_i x_i(v_i, \mathbf{v}_{-i}) - p_i(v_i, \mathbf{v}_{-i}) \;\ge\; 0.
    \label{eq:ir}
  \end{equation}
  No bidder ever receives negative utility from participating truthfully.
\end{enumerate}

By the standard monotonicity characterization of DSIC mechanisms, a deterministic allocation rule $x_i: V^n \to \{0,1\}$ is implementable if and only if $x_i(\cdot, \mathbf{v}_{-i})$ is monotone non-decreasing in $v_i$: once bidder $i$ wins at threshold $t_i(\mathbf{v}_{-i})$, bidder $i$ must win for all $v_i \ge t_i(\mathbf{v}_{-i})$ and pays the threshold price $p_i(v_i, \mathbf{v}_{-i}) = t_i(\mathbf{v}_{-i})$.

\subsection{The Three-Bidder APX-Hardness Theorem}

Let \prob{3-Optimal-Auction} denote the problem of computing a deterministic DSIC, IR allocation and payment rule that maximizes expected revenue~\eqref{eq:auction-revenue} given an explicit 3-dimensional probability array $\varphi(x, y, z)$ over $V^3$.

\begin{theorem}[Papadimitriou and Pierrakos \cite{Papadimitriou2011}]
\label{thm:3-optimal-auction}
\prob{3-Optimal-Auction} is $\APX$-hard. In particular, it is $\NP$-hard to approximate the maximum expected revenue of a deterministic DSIC, IR three-bidder auction within a factor smaller than $104/103 - \varepsilon$ for any $\varepsilon > 0$.
\end{theorem}

\begin{proof}[Proof outline]
The reduction in \cite{Papadimitriou2011} proceeds in three geometric and combinatorial stages:
\begin{enumerate}[label=\textbf{Step \arabic*.},leftmargin=*,topsep=2pt,itemsep=3pt]
  \item \textbf{Marginal revenue decomposition on the 3D valuation grid.}
  Fix the discrete 3D valuation grid $V^3$ with joint distribution $\varphi(x, y, z)$. Because $x_1(\cdot, y, z)$ is monotone in $x$, allocating the item to bidder $1$ along a fixed line $(y, z)$ corresponds to choosing a threshold $x$ and allocating the item to bidder $1$ on the upward ray $\{(x', y, z) : x' \ge x\}$, collecting revenue $x \sum_{x' \ge x} \varphi(x', y, z)$. Define the marginal revenue contribution $f(x, y, z)$ inductively from the top of the grid ($x = 1$ down to $0$) by $f(1, y, z) = \varphi(1, y, z)$ and
  \[
    f(x, y, z) = \max\Biggl\{\, x \sum_{x' \ge x} \varphi(x', y, z) - \sum_{x' > x} f(x', y, z),\; 0 \,\Biggr\},
  \]
  with analogous functions $g(x, y, z)$ and $h(x, y, z)$ along the $y$- and $z$-axes for bidders $2$ and $3$. This converts the revenue objective into an additive weight over axis-aligned upward rays.

  \item \textbf{Reduction to \prob{3-Segment} (maximum-weight disjoint rays).}
  Define an $x$-\emph{segment} anchored at $(x, y, z)$ (where $f(x,y,z) > 0$) to be the upward ray $\{(x', y, z) : x' \ge x\}$ with weight $w_x(x,y,z) = \sum_{x' \ge x} f(x', y, z)$, and define $y$-segments and $z$-segments analogously using $g$ and $h$. Because the item can be allocated to at most one bidder at each valuation profile $(x, y, z)$, feasible deterministic monotone allocations correspond to collections of pairwise non-intersecting $x$-, $y$-, and $z$-segments in the 3D grid. For $n = 2$ bidders, the intersection graph of horizontal $x$-segments and vertical $y$-segments is bipartite, so finding a maximum-weight independent set of segments reduces to minimum $s$-$t$ cut and is solvable in polynomial time. For $n = 3$ bidders, segments in three orthogonal directions form a 3-partite intersection family (\prob{3-Segment}), which is L-reducible to \prob{3-Optimal-Auction}.

  \item \textbf{L-reduction from \prob{Max-3CAT-SAT} to \prob{3-Segment}.}
  Define \prob{3CAT-SAT} to be the restriction of \prob{3SAT} in which variables are partitioned into three categories $\{x_i\}$, $\{y_j\}$, and $\{z_k\}$, and every clause contains at most one literal from each category. By a PCP-based reduction, \prob{Max-3CAT-SAT} is $\NP$-hard to approximate within $104/103 - \varepsilon$. Embedding variable gadgets and clause gadgets as orthogonal segments in $V^3$ (together with heavy scaffold segments that enforce directional consistency and prevent $v_i$ and $\neg v_i$ from being selected simultaneously) yields an L-reduction $\prob{Max-3CAT-SAT} \le_L \prob{3-Segment} \le_L \prob{3-Optimal-Auction}$.\qedhere
\end{enumerate}
\end{proof}

From an economic standpoint, Theorem~\ref{thm:3-optimal-auction} reveals a sharp computational threshold between $n = 2$ and $n = 3$ bidders under correlated priors: whereas two-bidder optimal deterministic auctions can be computed exactly in polynomial time, three-bidder optimal deterministic auctions do not admit a $\PTAS$ unless $\Pclass = \NP$ (although randomized Lotteries over menus can still be optimized via linear programming in the size of the support).

\section{Algorithmic Paradigms Across the APX Boundary: LP Rounding, SDP, and Primal-Dual}
\label{sec:algorithms}

We now turn from hardness reductions to the algorithmic techniques that place problems inside $\APX$ or $\logAPX$: linear programming relaxation with randomized rounding, semidefinite programming with geometric hyperplane rounding, and greedy/primal-dual methods.

\subsection{Randomization in Approximation Algorithms}

Randomization plays two distinct roles in complexity theory and approximation. For decision problems, randomized polynomial-time algorithms (such as the Miller--Rabin primality test \cite{Rabin1980} in $\coRP$, later derandomized deterministically by Agrawal, Kayal, and Saxena \cite{Agrawal2004}) use random coin flips to decide membership with bounded error, and a central conjecture in complexity theory is that $\BPP = \Pclass$. In approximation algorithms, randomization is frequently used as a \emph{rounding mechanism}: one first solves a continuous convex relaxation (an LP or SDP) in deterministic polynomial time to obtain fractional variables $x^* \in [0, 1]^m$, and then interprets $x^*$ as a probability distribution from which an integral solution is sampled \cite{Raghavan1987}. Because expectations are linear and sums of independent bounded random variables concentrate sharply around their expectations by the Chernoff bound, randomized rounding often yields clean approximation guarantees that can subsequently be derandomized into deterministic $\APX$ or $\logAPX$ algorithms using the method of conditional expectations or pairwise independent hash families \cite{Vazirani2001,WilliamsonShmoys2011}.

\subsection{Multicommodity Congestion Minimization via LP Relaxation and Randomized Rounding}
\label{subsec:congestion}

Consider the directed \emph{Congestion Minimization} problem \cite{Raghavan1987,WilliamsonShmoys2011}: we are given a directed graph $G = (V, E)$ with positive integer edge capacities $c_e \in \mathbb{Z}_{>0}$ for each $e \in E$, and $k$ terminal pairs $(s_1, t_1), \dots, (s_k, t_k) \in V \times V$ such that $G$ contains at least one directed $s_i \to t_i$ path for each $i \in \{1, \dots, k\}$. We must select one directed $s_i \to t_i$ path $P_i$ for each commodity $i \in \{1, \dots, k\}$. Given $P_1, \dots, P_k$, the load on edge $e \in E$ is $\ell_e = |\{i \in \{1,\dots,k\} : e \in P_i\}|$, and the congestion of edge $e$ is $\ell_e / c_e$. The goal is to minimize the maximum edge congestion $\max_{e \in E} (\ell_e / c_e)$, which is $\NP$-hard by reduction from \prob{Edge-Disjoint-Paths}.

To construct an LP relaxation, introduce a fractional flow variable $x_{i,e} \in [0, 1]$ for each commodity $i \in \{1, \dots, k\}$ and edge $e \in E$ (which is $1$ in an integral solution if $e \in P_i$ and $0$ otherwise), together with a scalar variable $r \ge 0$ representing the maximum congestion. Let $\delta^+(v)$ and $\delta^-(v)$ denote the sets of directed edges leaving and entering vertex $v \in V$. The fractional multicommodity flow linear program is:
\begin{equation}
\begin{alignedat}{3}
  &\text{minimize}   &\quad & r \\
  &\text{subject to} &\quad & \sum_{e \in \delta^+(v)} x_{i,e} - \sum_{e \in \delta^-(v)} x_{i,e} = \begin{cases}
    1  & \text{if } v = s_i, \\
    -1 & \text{if } v = t_i, \\
    0  & \text{if } v \in V \setminus \{s_i, t_i\},
  \end{cases} &\quad &\forall i \in \{1,\dots,k\},\; \forall v \in V, \\
  &                  &\quad & \sum_{i=1}^k x_{i,e} \le c_e \cdot r, &\quad &\forall e \in E, \\
  &                  &\quad & x_{i,e} \ge 0, \quad r \ge 0, &\quad &\forall i \in \{1,\dots,k\},\; \forall e \in E.
\end{alignedat}
\label{eq:congestion-lp}
\end{equation}
Any feasible routing $(P_1, \dots, P_k)$ induces a feasible solution to~\eqref{eq:congestion-lp} with $r = \max_{e \in E}(\ell_e / c_e)$, so the optimal LP objective value $r^*$ satisfies
\begin{equation}
  r^* \le \OPT_{\prob{Congestion}}(G).
  \label{eq:rstar-lower}
\end{equation}
Furthermore, by scaling capacities if necessary, we may assume $c_e r^* \ge 1$ for every edge $e$ used by at least one path (since in any integral solution each used edge carries load at least $1$, so $\OPT(G) \ge 1/c_e$).

\begin{algorithm}[t]
\caption{Flow-to-Path Decomposition Algorithm}
\label{alg:path-decomp}
\begin{algorithmic}[1]
\Require Directed graph $G = (V, E)$, terminal pair $(s_i, t_i)$, acyclic unit $s_i \to t_i$ flow $\{x_{i,e}\}_{e \in E}$.
\Ensure Probability distribution $\mathcal{D}_i = \{(P, \pi_P)\}$ over directed $s_i \to t_i$ paths.
\State Initialize $\mathcal{D}_i \leftarrow \emptyset$
\While{there exists a directed path $P$ from $s_i$ to $t_i$ with $x_{i,e} > 0$ for all $e \in P$}
  \State $\pi_P \leftarrow \min_{e \in P} x_{i,e}$
  \State $\mathcal{D}_i \leftarrow \mathcal{D}_i \cup \{(P, \pi_P)\}$
  \ForAll{$e \in P$}
    \State $x_{i,e} \leftarrow x_{i,e} - \pi_P$
  \EndFor
\EndWhile
\State \Return $\mathcal{D}_i$
\end{algorithmic}
\end{algorithm}

The Raghavan--Thompson randomized rounding algorithm \cite{Raghavan1987,WilliamsonShmoys2011} solves LP~\eqref{eq:congestion-lp} in polynomial time and rounds the fractional flow $(x_{i,e})$ in two steps:
\begin{enumerate}[label=\textbf{Step \arabic*.},leftmargin=*,topsep=2pt,itemsep=2pt]
  \item \textbf{Cycle elimination and path decomposition.}
  For each commodity $i \in \{1, \dots, k\}$, the flow conservation constraints in~\eqref{eq:congestion-lp} state that $\{x_{i,e}\}_{e \in E}$ is a unit $s_i \to t_i$ flow. If there is a directed cycle $C$ on which $x_{i,e} > 0$ for all $e \in C$, let $\delta = \min_{e \in C} x_{i,e}$ and subtract $\delta$ from $x_{i,e}$ along $C$; this preserves flow conservation, weakly decreases $\sum_i x_{i,e}$ on every edge, and zeroes out at least one edge variable. Repeating at most $m = |E|$ times yields an \emph{acyclic} unit flow $\{x_{i,e}\}$. Running Algorithm~\ref{alg:path-decomp} on this acyclic flow zeroes out at least one edge in every iteration of the $\mathbf{while}$ loop, terminating in at most $m$ iterations with a collection of paths $\mathcal{D}_i = \{(P, \pi_P)\}$ such that
  \begin{equation}
    \sum_{(P, \pi_P) \in \mathcal{D}_i} \pi_P = 1
    \qquad\text{and}\qquad
    \sum_{(P, \pi_P) \in \mathcal{D}_i :\, e \in P} \pi_P = x_{i,e}
    \quad \forall e \in E.
    \label{eq:path-marginal}
  \end{equation}
  \item \textbf{Independent path sampling.}
  For each commodity $i \in \{1, \dots, k\}$ independently, sample one path $P_i$ from $\mathcal{D}_i$ according to the probabilities $\pi_P$.
\end{enumerate}

\begin{theorem}[Raghavan--Thompson Congestion Bound \cite{Raghavan1987}]
\label{thm:congestion-bound}
On a directed graph with $m = |E| \ge 2$ edges, independent sampling from the path distributions $\mathcal{D}_1, \dots, \mathcal{D}_k$ produces paths $P_1, \dots, P_k$ such that with probability at least $1/2$,
\begin{equation}
  \max_{e \in E} \frac{\ell_e}{c_e} \;\le\; \alpha \cdot \OPT_{\prob{Congestion}}(G),
  \qquad\text{where}\quad
  \alpha = \frac{3 \ln(2m)}{\ln \ln(2m)} = O\!\left(\frac{\log m}{\log \log m}\right).
\end{equation}
\end{theorem}
\begin{proof}
Fix an edge $e \in E$. For each commodity $i \in \{1, \dots, k\}$, let $Y_{i,e} = \mathbf{1}[e \in P_i] \in \{0, 1\}$ be the indicator that the sampled path $P_i$ traverses $e$, and define the scaled random variable
\[
  X_{i,e} = \frac{Y_{i,e}}{c_e r^*} \in [0, 1].
\]
By~\eqref{eq:path-marginal}, $\Pr[Y_{i,e} = 1] = x_{i,e}$, and because $P_1, \dots, P_k$ are sampled independently, $X_{1,e}, \dots, X_{k,e}$ are independent random variables in $[0, 1]$. Their sum equals the normalized congestion on edge $e$:
\[
  \sum_{i=1}^k X_{i,e} = \frac{\ell_e}{c_e r^*},
  \qquad\text{with expectation}\qquad
  \mu_e = \mathbb{E}\!\left[\sum_{i=1}^k X_{i,e}\right] = \sum_{i=1}^k \frac{x_{i,e}}{c_e r^*} \le 1
\]
by the second constraint of LP~\eqref{eq:congestion-lp}. By the multiplicative Chernoff bound for sums of independent random variables in $[0, 1]$ with mean $\mu_e \le 1$, for any $\alpha > e$,
\[
  \Pr\!\left[\sum_{i=1}^k X_{i,e} \ge \alpha\right]
  \le \left(\frac{e}{\alpha}\right)^\alpha.
\]
Choosing $\alpha = \frac{3\ln(2m)}{\ln\ln(2m)}$ gives $(e/\alpha)^\alpha \le \frac{1}{2m}$. Taking a union bound over all $m$ edges $e \in E$,
\[
  \Pr\!\left[\exists e \in E : \frac{\ell_e}{c_e} \ge \alpha r^*\right]
  \le \sum_{e \in E} \frac{1}{2m} = \frac{1}{2}.
\]
With probability at least $1/2$, every edge $e \in E$ satisfies $\ell_e / c_e \le \alpha r^* \le \alpha \OPT_{\prob{Congestion}}(G)$ by~\eqref{eq:rstar-lower}.
\end{proof}

\subsection{Max-Cut: From Randomized Coin Flips to Semidefinite Programming}
\label{subsec:maxcut}

We proved in Theorem~\ref{thm:max-cut} that \prob{Max-Cut} is $\APX$-complete. Here we compare two randomized rounding paradigms on a weighted undirected graph $G = (V, E)$ with non-negative edge weights $w_e \ge 0$ for $e \in E$.

\paragraph{Uniform random partitioning ($0.5$-approximation).}
Assign each vertex $v \in V$ independently to $A_1$ with probability $1/2$ and to $A_2 = V \setminus A_1$ with probability $1/2$. Each edge $e = (u, v) \in E$ crosses the cut $(A_1, A_2)$ with probability $\Pr[e \in E(A_1, A_2)] = 1/2$. By linearity of expectation, the expected weight of the cut is
\begin{equation}
  \mathbb{E}\!\left[\sum_{e \in E(A_1, A_2)} w_e\right]
  = \sum_{e \in E} w_e \Pr[e \in E(A_1, A_2)]
  = \frac{1}{2} \sum_{e \in E} w_e
  \ge \frac{1}{2}\OPT_{\prob{Max-Cut}}(G).
  \label{eq:maxcut-half}
\end{equation}
Because the event $e \in E(A_1, A_2)$ depends only on the two endpoints of $e$, this expectation is unchanged if the vertex assignments are drawn from a pairwise independent hash family (requiring only $O(\log n)$ random bits, so all $O(n)$ seeds can be enumerated deterministically), or derandomized sequentially via conditional expectations (greedily placing each vertex on whichever side maximizes the cut edges to already-placed vertices).

\paragraph{Goemans--Williamson SDP relaxation ($0.87856$-approximation).}
To improve upon the factor $1/2$, Goemans and Williamson \cite{Goemans1995} formulated \prob{Max-Cut} as a quadratic integer program and relaxed it to a semidefinite program (SDP):
\begin{enumerate}[label=\textbf{\arabic*.},leftmargin=*,topsep=2pt,itemsep=2pt]
  \item \textbf{Exact quadratic formulation:} Associate each vertex $i \in V$ with a spin $y_i \in \{-1, +1\}$, so that $\frac{1 - y_i y_j}{2}$ equals $1$ if $(i, j)$ crosses the cut and $0$ otherwise:
  \begin{equation}
    \OPT_{\prob{Max-Cut}}(G) = \max_{y \in \{-1,+1\}^n} \frac{1}{2} \sum_{(i,j) \in E} w_{ij}(1 - y_i y_j).
    \label{eq:maxcut-qp}
  \end{equation}
  \item \textbf{Vector / SDP relaxation:} Relax each scalar $y_i \in \{-1, +1\}$ to a unit vector $u_i \in \mathbb{S}^{n-1} \subset \mathbb{R}^n$ with $\|u_i\|_2 = 1$:
  \begin{equation}
    \mathrm{SDP}(G) = \max_{\substack{u_1, \dots, u_n \in \mathbb{R}^n \\ \|u_i\|_2^2 = 1 \;\forall i}} \frac{1}{2} \sum_{(i,j) \in E} w_{ij}(1 - \langle u_i, u_j \rangle).
    \label{eq:maxcut-sdp}
  \end{equation}
  Writing the Gram matrix $Y_{ij} = \langle u_i, u_j \rangle$, problem~\eqref{eq:maxcut-sdp} is a linear optimization over the cone of symmetric positive semidefinite matrices $Y \succeq 0$ with $Y_{ii} = 1$, solvable to arbitrary additive precision in polynomial time via interior-point or ellipsoid methods, followed by Cholesky factorization $Y = U^\top U$ to recover $u_1, \dots, u_n \in \mathbb{S}^{n-1}$.
  \item \textbf{Random hyperplane rounding:} Draw a unit vector $r \in \mathbb{S}^{n-1}$ uniformly at random (equivalently, $r \sim \mathcal{N}(0, I_n)$) and set $S = \{i \in V : \langle r, u_i \rangle \ge 0\}$ and $T = V \setminus S$. Because the projection of $r$ onto $\operatorname{span}(u_i, u_j)$ is rotationally invariant, the probability that the hyperplane $r^\perp$ separates $u_i$ and $u_j$ is proportional to the angle $\theta_{ij} = \arccos(\langle u_i, u_j \rangle) \in [0, \pi]$ between them:
  \begin{equation}
    \Pr[(i,j) \in E(S, T)] = \frac{\arccos(\langle u_i, u_j \rangle)}{\pi}
    \ge \alpha_{\mathrm{GW}} \cdot \frac{1 - \langle u_i, u_j \rangle}{2},
    \label{eq:gw-angle}
  \end{equation}
  where
  \begin{equation}
    \alpha_{\mathrm{GW}} = \min_{0 < \theta \le \pi} \frac{2\theta}{\pi(1 - \cos\theta)} \approx 0.878567.
  \end{equation}
  Summing~\eqref{eq:gw-angle} over all edges $(i,j) \in E$ by linearity of expectation gives
  \[
    \mathbb{E}[\VAL(S, T)] \ge \alpha_{\mathrm{GW}} \cdot \mathrm{SDP}(G) \ge 0.87856\,\OPT_{\prob{Max-Cut}}(G).
  \]
\end{enumerate}
Khot, Kindler, Mossel, and O'Donnell \cite{Khot2007} subsequently proved that under the Unique Games Conjecture, no polynomial-time algorithm can approximate \prob{Max-Cut} within $\alpha_{\mathrm{GW}} + \varepsilon$ for any $\varepsilon > 0$.

\subsection{Minimum Set Cover: Greedy, Dual Fitting, and Randomized Rounding}
\label{subsec:setcover}

We conclude our algorithmic survey with \prob{Min-Set-Cover}, the canonical problem in $\logAPX$ \cite{Vazirani2001,WilliamsonShmoys2011}. Given a universe $U = \{e_1, \dots, e_n\}$ of $n$ elements, a family of subsets $\mathcal{S} = \{S_1, \dots, S_m\}$ with $\bigcup_{i=1}^m S_i = U$, and non-negative costs $c_i \in \mathbb{Q}_{>0}$ for each $S_i \in \mathcal{S}$, the goal is to select a subcollection $I \subseteq \{1, \dots, m\}$ with $\bigcup_{i \in I} S_i = U$ that minimizes total cost $\sum_{i \in I} c_i$.

\paragraph{LP relaxation and primal-dual / dual-fitting analysis.}
Introducing an indicator variable $x_i \in \{0, 1\}$ for whether set $S_i$ is chosen, the integer program and its linear programming relaxation (alongside the LP dual with variables $y_j \ge 0$ for each element $e_j \in U$) take the form:
\begin{equation}
\begin{aligned}
  &\textbf{Primal LP:} &&\min \sum_{i=1}^m c_i x_i &&\text{s.t.}\quad \sum_{i:\, e_j \in S_i} x_i \ge 1 \;\;\forall j \in [n], \quad x_i \ge 0 \;\;\forall i \in [m], \\
  &\textbf{Dual LP:}   &&\max \sum_{j=1}^n y_j     &&\text{s.t.}\quad \sum_{j:\, e_j \in S_i} y_j \le c_i \;\;\forall i \in [m], \quad y_j \ge 0 \;\;\forall j \in [n].
\end{aligned}
\label{eq:setcover-primal-dual}
\end{equation}
Two classic deterministic paradigms follow directly from~\eqref{eq:setcover-primal-dual} \cite{Vazirani2001}:
\begin{itemize}[topsep=2pt,itemsep=2pt]
  \item \textbf{Greedy algorithm via dual fitting ($H_n$-approximation):} In each step, select the set $S_i$ minimizing the price per newly covered element $c_i / |S_i \cap R|$, where $R$ is the set of currently uncovered elements, and charge each newly covered element $e_j \in S_i \cap R$ the price $p_j = c_i / |S_i \cap R|$. Ordering the elements of any set $S_i$ of size $k = |S_i|$ in the order they are covered, the $r$-th element of $S_i$ is covered when at least $k - r + 1$ elements of $S_i$ remain uncovered, so its price satisfies $p_{j_r} \le c_i / (k - r + 1)$. Summing over $r = 1, \dots, k$ shows $\sum_{e_j \in S_i} p_j \le c_i H_k \le c_i H_n$, where $H_n = \sum_{r=1}^n 1/r \le \ln n + 1$. Thus $y_j = p_j / H_n$ is Dual-LP feasible, and weak duality implies
  \[
    \ALG_{\mathrm{Greedy}} = \sum_{j=1}^n p_j = H_n \sum_{j=1}^n y_j \le H_n \OPT_{\mathrm{LP}} \le (\ln n + 1)\OPT.
  \]
  \item \textbf{Primal-dual / frequency-$f$ approximation:} If each element $e_j \in U$ appears in at most $f = \max_j |\{i : e_j \in S_i\}|$ sets, raising dual variables $y_j$ of uncovered elements until Dual constraints $\sum_{e_j \in S_i} y_j = c_i$ become tight (or taking all sets with $x_i^* \ge 1/f$ in an optimal primal solution $x^*$) yields an $f$-approximation. Because \prob{Min-Vertex-Cover} is the special case of \prob{Min-Set-Cover} in which each edge $e_j$ belongs to exactly $f = 2$ vertex stars, this recovers the $2$-approximation placing $\prob{Min-Vertex-Cover} \in \APX$.
\end{itemize}

\begin{algorithm}[t]
\caption{Randomized LP Rounding for \prob{Min-Set-Cover}}
\label{alg:setcover-rounding}
\begin{algorithmic}[1]
\Require Universe $U$ of size $n$, sets $\mathcal{S} = \{S_1, \dots, S_m\}$ with costs $c_i$, optimal fractional solution $x^* \in [0, 1]^m$ to the Primal LP in~\eqref{eq:setcover-primal-dual}, oversampling factor $d \ge 2$.
\Ensure Subcollection $\mathcal{C} \subseteq \mathcal{S}$.
\State Initialize $\mathcal{C} \leftarrow \emptyset$
\State Set number of rounds $T \leftarrow \lceil d \ln n \rceil$
\For{$t = 1$ \textbf{to} $T$}
  \State Initialize $\mathcal{C}^{(t)} \leftarrow \emptyset$
  \For{$i = 1$ \textbf{to} $m$}
    \State Include $S_i$ in $\mathcal{C}^{(t)}$ independently with probability $x_i^*$
  \EndFor
  \State $\mathcal{C} \leftarrow \mathcal{C} \cup \mathcal{C}^{(t)}$
\EndFor
\State \Return $\mathcal{C}$
\end{algorithmic}
\end{algorithm}

\paragraph{Randomized LP rounding ($O(\log n)$-approximation).}
Alternatively, let $x^* \in [0, 1]^m$ be an optimal solution to the Primal LP in~\eqref{eq:setcover-primal-dual} and run Algorithm~\ref{alg:setcover-rounding} for $T = \lceil d \ln n \rceil$ independent rounds \cite{Vazirani2001,WilliamsonShmoys2011}. In a single round $t$, the expected cost of $\mathcal{C}^{(t)}$ is $\sum_{i=1}^m c_i x_i^* = \OPT_{\mathrm{LP}} \le \OPT$, so by linearity of expectation across all $T$ rounds,
\begin{equation}
  \mathbb{E}[\operatorname{cost}(\mathcal{C})] \le \sum_{t=1}^T \mathbb{E}[\operatorname{cost}(\mathcal{C}^{(t)})] = T \cdot \OPT_{\mathrm{LP}} \le \lceil d \ln n \rceil \OPT.
  \label{eq:setcover-exp-cost}
\end{equation}
For any element $e_j \in U$, the primal constraint $\sum_{i:\, e_j \in S_i} x_i^* \ge 1$ and the inequality $1 - x \le e^{-x}$ imply that the probability $e_j$ is left uncovered in a single round $t$ is at most
\[
  \Pr[e_j \text{ uncovered in round } t]
  = \prod_{i:\, e_j \in S_i} (1 - x_i^*)
  \le \exp\!\Biggl(-\sum_{i:\, e_j \in S_i} x_i^*\Biggr)
  \le e^{-1}.
\]
Across $T = \lceil d \ln n \rceil$ independent rounds, $\Pr[e_j \text{ uncovered by } \mathcal{C}] \le e^{-T} \le n^{-d}$. By the union bound over all $n$ elements $e_j \in U$, the probability that $\mathcal{C}$ fails to cover $U$ is at most $n \cdot n^{-d} = n^{-(d-1)}$, and by Markov's inequality $\Pr[\operatorname{cost}(\mathcal{C}) \ge 4 T \OPT] \le 1/4$. For $d = 2$ and $n \ge 4$, Algorithm~\ref{alg:setcover-rounding} outputs a valid set cover of cost at most $4\lceil 2\ln n\rceil \OPT = O(\log n)\OPT$ with probability at least $1 - 1/n - 1/4 = 3/4 - 1/n > 1/2$. Dinur and Steurer \cite{DinurSteurer2014} proved that no polynomial-time algorithm can approximate \prob{Min-Set-Cover} within $(1 - \varepsilon)\ln n$ for any $\varepsilon > 0$ unless $\Pclass = \NP$, confirming that $\prob{Min-Set-Cover} \in \logAPX \setminus \APX$.

\section{Discussion and Open Directions}
\label{sec:discussion}

Studying $\NP$-hard optimization problems through the lens of $\APX$ and approximation-preserving reductions reveals structural distinctions that are invisible to exact decision complexity. Two problems that are isomorphic under Karp reductions---such as \prob{Min-Vertex-Cover} and \prob{Max-Independent-Set}, or \prob{2SAT} and \prob{Max-2SAT}---separate sharply once we ask how well polynomial-time algorithms can approximate their optima. We highlight three directions that connect the structural complexity of $\APX$ to ongoing research:

\paragraph{1. Query complexity and $\APX$ oracles.}
Krentel \cite{Krentel1988} introduced the class $\OptP$ of functions computable by taking the maximum output along accepting paths of an $\NP$ machine, showing that $\FP^{\NP}$ (polynomial time with $O(\poly(n))$ queries to an $\NP$ oracle) captures exact optimization for \prob{TSP} and \prob{0-1 Integer Programming}, whereas problems with polynomially bounded optima such as \prob{Max-Clique} require only $O(\log n)$ queries ($\FP^{\NP[\log n]}$). A natural relativized question is how much computational power is provided by an \emph{approximation oracle} $\mathcal{O}_{\Pi, c}$ that returns, for any instance $I$ of an $\NPO$ problem $\Pi$, a value (or solution) guaranteed to lie within a fixed factor $c$ of $\OPT_\Pi(I)$. When $\Pi \in \APX$ and $c$ is at least the algorithmic threshold of $\Pi$, $\mathcal{O}_{\Pi, c}$ can be simulated directly in $\FP$, so $\Pclass^{\mathcal{O}_{\Pi,c}} = \Pclass$. Conversely, by Corollary~\ref{cor:no-ptas} and Lemma~\ref{lem:pcp}, an oracle $\mathcal{O}_{\prob{Max-3SAT}, 1+\varepsilon}$ with $\varepsilon < 1/8$ solves every language in $\NP$ with a single query ($\Pclass^{\mathcal{O}_{\prob{Max-3SAT}, 1+\varepsilon}[1]} = \NP$). Characterizing the intermediate regime---for instance, whether non-adaptive queries to a $c$-approximation oracle for a $\logAPX$-complete problem such as \prob{Min-Set-Cover} or a $\polyAPX$-complete problem such as \prob{Max-Clique} collapse levels of the boolean hierarchy or polynomial hierarchy---connects hardness amplification directly to bounded-query complexity \cite{Crescenzi1999}.

\paragraph{2. Combinatorial boundaries between $\APX$ and $\logAPX$.}
For covering and packing problems of the form~\eqref{eq:setcover-primal-dual}, general instances lie in $\logAPX \setminus \APX$, yet many geometric and graph-theoretic restrictions fall inside $\APX$ or even $\PTAS$. Whenever an unweighted set system $(U, \mathcal{S})$ has bounded VC-dimension $d$, $\varepsilon$-net rounding improves the $O(\log n)$ approximation factor to $O(d \log(d \cdot \OPT))$ \cite{Bronnimann1995}, and when each element has bounded frequency $f$ primal-dual rounding achieves factor $f$ \cite{Vazirani2001,Garg1996}. Giving a necessary and sufficient combinatorial characterization of the hypergraph families for which \prob{Min-Set-Cover} belongs to $\APX$ remains an active problem.

\paragraph{3. Closing constant-factor gaps inside $\APX$.}
Even after a problem is shown to be $\APX$-complete, pinpointing its exact approximability constant $c^*$ often requires new algorithmic and PCP machinery. For \prob{Max-3SAT}, the $7/8$ threshold is unconditionally tight under $\Pclass \neq \NP$ \cite{Hastad2001}. For \prob{Max-Cut} ($0.87856$ vs.\ $16/17$) and \prob{Min-Vertex-Cover} ($2$ vs.\ $1.3606$), the algorithmic LP/SDP integrality gaps match the conditional Unique Games hardness bounds \cite{Khot2007,KhotRegev2008} but remain separated from the best known unconditional $\NP$-hardness bounds. Likewise, for \prob{Metric-TSP} ($3/2 - \varepsilon_0$ \cite{Karlin2021} vs.\ $123/122$ \cite{Karpinski2015}) and \prob{ATSP} \cite{Svensson2020}, narrowing the gap between the subtour elimination LP integrality gap ($4/3$ for \prob{Metric-TSP}) and polynomial-time rounding algorithms remains one of the central open problems in combinatorial optimization.

\section*{Acknowledgements}
The authors thank Li-Yang Tan, Tom Knowles, and Can Liu for helpful discussions and feedback on an earlier version of this work.

\bibliographystyle{plain}
\bibliography{refs}

\end{document}